\documentclass{amsart}
\usepackage{amsmath,oldlfont,amssymb} \normalsize 

\def\hyph{-\penalty0\hskip0pt\relax}
\def\fder#1{ \frac {d}{d \, {#1}}}
\def\dertn#1{ \frac {d^{#1}}{d \, t^{#1}}}
\def\dert{ \frac {d}{d \, t}}
\def\cW{{\ifmmode \cal W \else $\cal W $ \fi }}
\def\cC{{\ifmmode \cal C \else $\cal C $ \fi }}
\def\gm{{\ifmmode \mu \else $\mu $ \fi }}
\def\gb{{\ifmmode \beta \else $\beta $ \fi }}
\def\fM{{\ifmmode \mathfrak M \else $\mathfrak M$ \fi}}
\def\fG{{\ifmmode \mathfrak G \else $\mathfrak G$ \fi}}
\def\gd{{\ifmmode \delta \else $\delta $ \fi }}
\def\gl{{\ifmmode \lambda \else $\lambda $ \fi }}
\def\go{{\ifmmode \omega \else $\omega $ \fi }}
\def\gs{{\ifmmode \sigma \else $\sigma $ \fi }}
\def\ollim#1{\stackrel{O_{\ell^{#1}}}{\rightarrow}}
\def\gL{{\ifmmode \Lambda \else $\Lambda $ \fi }}
\def\gk{{\ifmmode \kappa \else $\kappa $ \fi }}
\def\gz{{\ifmmode \zeta \else $\zeta $ \fi }}
\def\cI{{\ifmmode \cal I \else $\cal I $ \fi }}
\def\fD{{\ifmmode \mathfrak D \else $\mathfrak D$ \fi}}
\def\fT{{\ifmmode \mathfrak T \else $\mathfrak T$ \fi}}
\def\fg{{\ifmmode \mathfrak g \else $\mathfrak g$ \fi}}
\def\fh{{\ifmmode \mathfrak h \else $\mathfrak h$ \fi}}
\def\fl{{\ifmmode \mathfrak l \else $\mathfrak l$ \fi}}
\def\fQ{{\ifmmode \mathfrak Q \else $\mathfrak Q$ \fi}}
\def\gr{{\ifmmode \rho \else $\rho $ \fi }}
\def\gt{{\ifmmode \tau \else $\tau $ \fi }}
\def\ga{{\ifmmode \alpha \else $\alpha$ \fi }}
\def\BZ{{\ifmmode \mathbb Z\else $\mathbb Z$ \fi }}
\def\BN{{\ifmmode \mathbb N\else $\mathbb N$ \fi }}
\def\BR{{\ifmmode \mathbb R\else $\mathbb R$ \fi }}
\def\cT{{\ifmmode \cal T \else $\cal T $ \fi }}
\def\cL{{\ifmmode \cal L \else $\cal L $ \fi }}
\def\cA{{\ifmmode \cal A \else $\cal A $ \fi }}
\def\cD{{\ifmmode \cal D \else $\cal D $ \fi }}
\def\gth{{\ifmmode \theta \else $\theta $ \fi }}
\def\cB{{\ifmmode \cal B \else $\cal B $ \fi }}
\def\gn{{\ifmmode \nu \else $\nu $ \fi }}

\newcommand\supp{\operatorname{supp}}

\def\partder#1{\frac{\partial}{\partial #1}}
\def\df{\stackrel{\rm def}{=}}
\def\fR{{\ifmmode \mathfrak R \else $\mathfrak R$ \fi}}
\theoremstyle{remark}

\theoremstyle{definition}

\theoremstyle{plain}
\newtheorem{thm}{Theorem}[section]
\newtheorem{prop}[thm]{Proposition}
\newtheorem{lem}[thm]{Lemma}

\numberwithin{equation}{section}

\newtheorem{remark}[equation]       {Remark}
\begin{document}
\begin{titlepage}
\author{Christopher J. Winfield}
\address{University of Alaska - Fairbanks\\
Fairbanks, AK } \email{cjwinfield01@alaska.edu} \pagestyle{myheadings}
\markboth{Draft: C. Winfield}{Report}
\end{titlepage}
\begin{abstract}
\title{Continuum Eigenmodes in Some Linear Stellar Models}
\begin{sloppypar}
We apply parallel approaches in the study of continuous spectra to
adiabatic stellar models. We seek continuum eigenmodes for the LAWE
formulated as both finite difference and linear differential
equations. In particular, we apply methods of Jacobi matrices and
methods of subordinancy theory in these respective formulations. We
find certain pressure-density conditions which admit
positive-measured sets of continuous oscillation spectra under
plausible conditions on density and pressure. We arrive at results
of unbounded oscillations and computational or, perhaps, dynamic
instability.
\end{sloppypar}
\end{abstract}
 \maketitle
\pagenumbering{arabic}
\section*{Introduction}

The problem of radial, adiabatic (isentropic) stellar oscillations
is well studied in cases where discrete eigenmodes are calculated in
the frameworks of linearized differential equations and in parallel
applications of operator theory (see, for instance,
\cite{b1,b2,co,svh}).\footnote{In the general literature (cf.
\cite{co}), the problem is often referred to as the LAWE - the
linear adiabatic wave equation.} However, continuous eigenvalues,
although previously suggested in the general literature, are
typically absent, either excluded from particular physical models or
explicitly avoided to assure dynamic or numerical stability
\cite{bs,ko,lw,sterne} (see also the Appendix of \cite{c}). In this
article we investigate some cases where continuous spectra are
present which have ramifications on the associated motion - and
perhaps on the perturbation method itself. Our study involves
linearized perturbations of the differential system
\begin{align}
\dertn{2}{r} &=\frac{-G {\mathfrak M}(r)}{r^2}- 4\pi r^2\frac{\partial P}{\partial m}\notag\\
\frac{1}{\rho}&=\frac{4\pi}{3}\frac{\partial r^3}{\partial
m}.\label{sys}\end{align}
 Here, $m={\mathfrak M}(r)$ is total mass measured from the origin to
distance $r$ from the center of a spherically symmetric star, and
$P$ and $\rho$ are pressure and density, respectively, depending on
$r$. Moreover, the differential mass $dm$ is interpreted as that of
a spherical mass shell (or mass element) with mean radius $r$ about
the origin.
 We suppose, $0<r\leq R_*$, $0\leq \fM\leq \fM_*$ where
$R_*$ and $\fM_*$ are fixed (arbitrary) positive constants; and, we
consider only shell models where $r$ is confined to intervals of the
form $[R_*-\gd,R_*]$ for some fixed, but arbitrary, $0<\gd<R_*.$

In our study we adopt a background model where hydrostatic
equilibrium (HSE) holds approximately, which of course is a standing
assumption of the LAWE. We consider various adiabatic, polytropic
and non-polytropic, equations of state (EOS): By '{\em polytropic}'
we mean that pressure is a power function of density (not to be
confused with 'polytrope', a solution of the Lane-Emden equation).
We find that rigorous analysis has indeed been done \cite{b2,bs}
along these lines with conditions that ensure pure-point spectra
and, hence, discrete radial motion analogous to standing waves.
These conditions include those where HSE holds exactly on
neighborhoods of shell boundaries. Most of our models likewise
impose some exact HSE conditions, but, in contrast, only as a
boundary condition at the outer surface. For our classes of
polytropic EOS we instead obtain continuous spectra for operators
defined on radial intervals away from the origin with mass density
functions satisfying $\rho(r)$ $\propto(R_*-r)^a,$ $a>1$ (in stark
contrast to Section 4, \cite{sterne} which places extreme mass
concentration at the stellar center).

The general approach of this article is as follows: Our study of
eigenvalues first arises after linearization and discretization of
the system (\ref{sys}) whereby a change of variables converts a
finite difference equation into an eigenvalue problem
$(A-\gl)\vec{X}=\vec{0},$ a problem typically left to computation
involving large but finite matrices $A$ (cf. \cite{c,co}). We find,
however, that the adiabatic case admits tri-diagonal $A$'s which
lend themselves to analysis of Jacobi matrices after passing to an
infinite system: Here, we impose non-uniform partitions which we let
cluster at the stellar surface $r=R^*$. In this context one turns to
a large body of mathematical work on the infinite (cf. \cite{ks,s2})
Jacobi matrix. Indeed, cases of continuous spectra as well as
discrete spectra are well known whereby unbounded oscillations modes
can be discerned in this application by the nature of the spectra.
We then pass to second-order ordinary differential equation models,
applying various results arising from subordinancy theory
\cite{gp,p} and (Weyl) limit-point case Sturm-Liouville operators
\cite{cl, s, st}. The author is not aware if such oscillation
behavior is actually observed in stars. Perhaps, the properties
studied here are merely artifacts of the perturbation methods and
yet prove absent in more complete, non-linear methods or in simple
models which (say) incorporate non-adiabatic or stochastic
processes.

This article is organized as follows: In Section ~\ref{sc1} we
formulate our finite difference equation and place it in a
Jacobi-matrix format. In Section ~\ref{sc2} we apply results from
certain Hilbert-Schmidt and trace-class operators to arrive at
general results on spectra and solutions. In Section ~\ref{sc4} we
study non-polytropic cases where absolutely continuous spectra are
present and further develop some cases where pure-point spectra also
appear. In Section ~\ref{sc5} we reformulate and study the matrix
equation in some polytropic cases. Then, in Section ~\ref{sc6} we
present the differential equation form of the LAWE in
Sturm-Liouville (SL) form. Finally, in Sections ~\ref{sc7} and
~\ref{sc8} we revisit polytropic and non-polytropic cases,
respectively, in the SL context.
\section{The Difference Equation}\label{sc1}
To begin this section we outline a derivation of our finite
difference equation and the perturbation scheme involved which
follow from \cite{c, co}\footnote{For simplicity, we do not
precisely follow their mass-division numerical schemes (see also
\cite{hfg}) which is irrelevant to our study, given monotonically
decreasing $M(I)$.}. For the system (\ref{sys}) we suppose that $m=$
$\fM(r)$ is a strictly increasing function and treat $m$ as an
independent variable. Then, one replaces $r,P$ and $\rho$ by the
perturbed quantities $r(m)+\gd r(m,t),$ $P(r,\rho)+\gd P(r,t),$ and
$\rho(r)+\gd\rho(r,t),$ respectively. One perturbs the system about
HSE where
\begin{equation} 4\pi r^2\frac{\partial P}{\partial
m}=-\frac{G \fM(r)}{r^2}\label{pressure}\end{equation} and replace
$\gd r(m,t)\rightarrow e^{i\go t}\gd r(m):$ Here, $\partder{m}$
$\df$ $\frac{1}{4\pi r^2\rho(r)}\partder{r}.$ Upon discretization,
we introduce mass elements $M(I)\df\fM(I)-\fM(I-1)$ and set
$\fM(I)\df\sum_{j=1}^IM(j)$ (vacuous sums are zero) and, after a
separation of variables, obtain
\begin{equation}- \go^2 \gd r(I) =\frac{4G\fM(I)}{r^3(I)}\gd r(I)
-4\pi r^2(I)\left(\frac{\gd P(I)-\gd
P(I-1)}{M(I)}\right)\label{mainprob}\end{equation} where
\begin{align}
\gd P(I)&=(\Gamma P)(I)\cdot(\fR_1(I) X(I) - \fR_2(I)X(I+1)),\notag\\
\fR_{i}(I)&\df \left(\frac{\rho(I)}{{M}(I)}\right)
\left( \frac{4\pi r^2(I+i-1)}{\sqrt{M(I+i-1)}}\right):i=1,2\notag,\\
\Gamma &\df\left(\frac{\partial \ln P}{\partial\ln \rho}\right)_S
\text{(entropy $S$ held fixed)\,and,}\notag\\
X(I) &\df\gd r(I)\sqrt{M(I)}.\notag\end{align} By setting
\begin{align}
\gd P(I)&-\gd P(I-1)=((\Gamma P \fR_1)(I)+(\Gamma P \fR_2)(I-1))\cdot X(I)\notag\\
-&(\Gamma P\fR_2)(I)\cdot X(I+1)-(\Gamma P\fR_1)(I-1)\cdot
X(I-1)\notag\end{align} and adopting some notation of \cite{c}, we
arrive at our eigenvalue problem
\begin{equation}\gl X(I) = G_1(I)X(I-1)+G_2(I)X(I)+G_3(I)X(I+1)
\label{main}\end{equation} where $\gl\df-\go^2\leq 0,$
\begin{align}
G_3(I)&\df 16\pi^2\left(\frac{r^2}{\sqrt{M}}\right)(I)\left(\frac{\Gamma P\rho}{{M}}\right)(I)\left(\frac{r^2}{\sqrt{M}}\right)(I+1)\label{Gs}\\
&=16\pi^2\left(\frac{\Gamma P\rho r^4}{M^2}\right)(I)\frac{r^2(I+1)}{r^2(I)}\sqrt{\frac{M(I)}{M(I+1)}},\notag\\
G_1(I)&\df
G_3(I-1): I\geq 2,\text{\rm\,and}\notag\\
-G_2(I)&+\left(\frac{4G\fM}{r^3}\right)(I)\df
\frac{4\pi^2[(\Gamma P \fR_1)(I)-(\Gamma P\fR_2)(I-1)]r^2(I)}{M(I)}\notag\\
=&\left(\left(\frac{\Gamma P\rho}{{M}}\right)(I) +\left(\frac{\Gamma
P\rho}{{M}}\right)(I-1) \right)\left(\frac{4\pi
r^2}{\sqrt{M}}\right)^2(I)
\notag\\
=&G_3(I)\frac{r^2(I)\sqrt{M(I+1)}}{r^2(I+1)\sqrt{M(I)}}+
G_1(I)\frac{r^2(I)\sqrt{M(I-1)}}{r^2(I-1)\sqrt{M(I)}}\notag
\end{align}
In our sign convention, $\gl>0$ would correspond to exponential time
dependence.\footnote{We do not elaborate on the interpretation of
negative $\go^2$ but we defer to discussions found in Sections 8.7
and 8.8 of \cite{co}.}

We note that this model is also simplified by an assumption of
constant entropy and that the equations decouple from dependence on
temperature, opacity, and luminosity, the admission of which would
lead to larger systems and much greater complexity (cf. Section 9,
\cite{kw}). Moreover, we note that $P$ can represent pressure from
various sources - not simply mechanical - and can involve
temperature and chemical variation. Indeed, such effects could
conceivably be incorporated in a quasi-static EOS with $P$ having
dependence on $r$ as well as $\rho$.
\section{Matrix Works and Spectra}\label{sc2}
We may reformulate the difference equation (\ref{main}) in matrix
form as $A\vec{X}=-\go^2\vec{X}$
\begin{equation}A=\left(\begin{matrix} a_1&c_1&0&0&\hdots\\
c_1& a_2 & c_2& 0 &\hdots\\
0 & c_2& a_3& c_3&\hdots\\
0&0&c_3&a_4&\ddots\\
\vdots&\vdots&\vdots &\ddots
&\ddots\end{matrix}\right)\label{mtx}\end{equation}
$$c_I= G_3(I), \,\, a_I=G_2(I)$$
where we set $\vec{X}=$ $(X(1),X(2),\hdots)$ for $X(I)=$ $\gd r(I)
\sqrt{M(I)}$. More generally, we let $\vec{X}$ denote the function
(or vector) defined on $\mathbb{Z}^+$ given by the corresponding
assignments (components) $X(I):$ $I=1,2,\hdots$ and likewise define
$\vec{P}, \vec{\gd r},$ $\vec{\rho},$ etc. In our analysis we
consider $A$ as a bounded, self-adjoint operator on the space of
vectors $\vec{X}\in$ $\ell^2(\mathbb{Z}^+)$, noting that
$\supp_{j}|a_j|,$ $\supp_{i}|c_j|$ $<\infty.$ In turn, we may regard
$\vec{\gd r}$ as an element of a weighted $\ell^2$ space since
$$\sum_{I=1}^{\infty}(X(I))^2=\sum_{I=1}^{\infty}(\gd r(I))^2M(I)$$
where, for a star of finite mass, we have $\fM^*$ $\df$
$\sum_{I=1}^{\infty}M(I)$ $<\infty$. What we then find from a simple
$\ell^{p}$-space fact is the following
\begin{prop} \label{basicprop}
For a finite-mass star, the perturbation ${\gd r} (I)$ is unbounded
as $I\rightarrow \infty$ if the associated solution $\vec{X}$ is not
in $\ell^2({\mathbb Z}^+)$.
\end{prop}

For ease of exposition we suppose for now that for some real
constant $z$
\begin{equation}
\lim_{I\rightarrow \infty}G_2(I)= z;\,\,\, \lim_{I\rightarrow\infty}
G_3(I) = 1. \label{req}\end{equation} Now, let $J_0$ denote the
tri-diagonal matrix in the form (\ref{mtx}) but with $c_i=1, a_i=0$
$\forall i$ and let constants denote scalar multiplication. It is
immediate that $A-z-J_0$ is compact, whereby the essential spectrum
of $A$ is that of $z+J_0$ (see \cite{s2}).  More precise information
is obtained in the literature depending on whether
 $A-z-J_0$ is of Hilbert-Schmidt class, whereby
\begin{equation}
\label{hs} \sum_{I=1}^{\infty}\left|G_2(I)-z\right|^2+
\sum_{I=1}^{\infty}\left|G_3(I)-1\right|^2<\infty,\end{equation} of
trace class, whereby \begin{equation} \label{trace}
\sum_{I=1}^{\infty}\left|G_2(I)-z\right|+
\sum_{I=1}^{\infty}\left|G_3(I)-1\right|<\infty,\end{equation} or of
slightly more rapid convergence, such as
\begin{equation}\label{case}
\sum_{I=1}^{\infty} I\cdot\left|G_2(I)-z\right|+ \sum_{I=1}^{\infty}
I\cdot\left|G_3(I)-1\right|<\infty.
\end{equation}

We note the following: The mode of convergence of $G_1(I), G_2(I)$
will depend on $P,$ $\rho$ and $M$ and how they are inter-related;
these conditions may be imposed via mass conservation together with,
as we will show, an equation of state. Such modes can imply various
spectral properties of $A$ which in turn have implications on
associated solutions $A\vec{X}$ $=\lambda\vec{X}$ (cf. \cite{gp,s}).
Although this approach does not typically yield precise estimates of
$X(I)$ (especially when mixtures of spectral types are possible),
asymptotic (Jost) estimates in some cases do indeed obtain and, in
turn, lead to estimates of $\gd r(I)$.

Our main results depend on the following well known results which we
list together in
\begin{thm} (Previously Known) Suppose $r(I):I=1,2,\hdots$ is a monotonically increasing sequence
where $\lim_{I\rightarrow\infty}r(I)=R$ $<\infty$ and let
$\vec{M}\in \ell^1(\mathbb{Z}^+)$ with $M(I)>0$ $\forall I$. Then,
\begin{itemize}
\item[i)] if (\ref{req}) holds, then $\gs_{ess}$, essential spectrum, satisfies
$$\sigma_{ess}(A)=[z-2,z+2];$$
\item[ii)] if (\ref{hs}) holds, then the pure-point spectra $E_k:k=1,2,\hdots$ such that $|E_k-z|>2$
(perhaps none, finitely or infinitely many) of $A$ satisfy
$$
\sum_k\left(||E_k-z|-2|\right)^{3/2}<\infty;
$$\\
\item[iii)]if (\ref{trace}) holds, then $\gs_{ac}(A)=$ $[z-2,z+2]$ and, the
pure point spectrum $\sigma_{pp}$ satisfies
$$\sigma_{pp}(A)\bigcap (z-2,z+2)=\emptyset$$ which is to say that no $\ell^2$ eigenvalues exist in $(z-2,z+2);$
\item[iv)] finally, if (\ref{case}) holds, then $z\pm 2\notin \sigma_{pp}(A)$ hold along with iii).
 \end{itemize}
\label{first} \end{thm}
 We make several remarks:
Statement i) follows from a routine application (see \cite{s2}) of
the Weyl Invariance theorem, noting that $A-z-J_0$ is compact.
Statement ii) is one of several criteria (see Theorem 1 \cite{ks} or
Theorem 1.10.1 \cite{s2}) which together are equivalent to $A-z-J_0$
being of Hilbert-Schmidt class. Statements iii) and iv) follow from
Theorem A.6 \cite{ds2} (see also \cite{s}). The absolutely
continuous part $f(x)dx$ of the implied spectral measure is
supported on $[z-2,z+2],$ the positivity of which is demonstrated by
a finite lower bound on a weighted Lebesgue integral of $\log
|f(x)|$ (the so-called Quasi-Szeg\"o condition); indeed, in case
iii) an even stricter integral condition (the Szeg\"o condition) on
$f(x)$ also holds (see \cite{ks, s2}).

We introduce notation to indicate various modes of convergence of
sequences: The expression $a(I)$ $\ollim{p}$ $L$ means that
$\lim_{I\rightarrow\infty}a(I)=L$ and that for $b(I)=a(I)-L$ the
sequence satisfies $\vec{b}$ $\in \ell^p(\mathbb{Z}^+).$ For
example, a geometric sequence $a(I)=C \eta^I$ with fixed $C$ and
$|\eta|<1$ satisfies $a(I)\ollim{1}$ $||\vec{a}||_{\ell^1}$. We will
need to establish
\begin{prop} For sequences $\vec{a}$ and $\vec{b}$
suppose for some $p\geq 1$ that $a(I)\ollim{p}L_a$ and
$b(I)\ollim{p}L_b$ for some constants $L_a,L_b.$ Then,
$a(I)b(I)\ollim{p}L_aL_b.$ Moreover, $f(a_I)\ollim{p} f(L_a)$ for
any Lipshitz function $f.$ \label{abprop}
\end{prop}
\begin{proof} If either $L_a$ or $L_b$ are zero, the first result is clear since both $\vec{a}$ and $\vec{b}$
are bounded. If both $L_a$ or $L_b$ are non-zero, then by scaling
arguments, it will suffice to prove the result for $L_a=L_b=1.$ We
write
\begin{equation}a(I)b(I)-1=(a(I)-1)(b(I)+1)+(b(I)-1)-(a(I)-1).\label{ab}\end{equation}
From standard $\ell^p$ inequalities we find that the left-hand side
of (\ref{ab}) defines a sequence in $\ell^p(\mathbb{Z}^+)$. The last
statement is clear since $|f(a_I)-f(L)|^p\leq  (c|a_I-L|)^p$ for
some constant $c\geq 0.$
\end{proof}

Our examples in the discrete\hyph systems case will involve mass
distributions of the form
\begin{equation} M(I+1)=\eta^{\gamma}M(I);\,\,\, r(I)=R_*\cdot(1-\eta^{I})\label{massdist}
\end{equation} for fixed $0<\eta<1$ and $\gamma>0,$ in which case
$$\fM(I)=(1-\eta^{\gamma})\fM_*\sum_{j=0}^I \eta^{\gamma j}=\fM_*(1-\eta^{\gamma(I+1)}).$$ Here $\gamma=1$ corresponds to constant density (to order $O(\eta^{I})$).
It will be convenient to introduce \footnote{See equation (8.20) and
following discussion \cite{c} for physical interpretation of
quantities $\go^2 \gL^*$.} \begin{equation}\gL_*\df G
\fM_*/R^3_*;\,\,\gk\df
(4+\gz)/(\eta^{-\gamma/2}+\eta^{\gamma/2})>0\label{gLkappa}\end{equation}
for $\gz> -4,$ amounting to scaling and translation parameters, to
state
\begin{thm}\label{thm2}
Given a mass distribution of the form (\ref{massdist}), suppose that
for some fixed $\gz>-4$ \begin{equation} \left(\frac{\Gamma P
\rho}{M^2}\right)(I)\ollim{2} \frac{1+\gz/4}{4
\pi^2R_*^4\cdot(1+\eta^{-\gamma})}\gL_*.\label{GPoverM2}\end{equation}
Then, the essential spectrum of $A$ is given by $\gs_{ess}(A)=\cI$
for
\begin{equation}\cI\df [(-\gz-2\gk)\gL_*,(-\gz+2\gk)\gL_*]\label{spect}\end{equation}
and the eigenvalues $\gl_k$ satisfy
\begin{equation}|\gl_k+\gz\gL_*|\ollim{3/2} 2\gk\gL_*\label{evallim}\end{equation} if there are infinitely many.
Moreover, if the convergence (\ref{GPoverM2}) is in $O_{\ell^1}$
mode, the solutions space corresponding to each $\gl$ $\in$
$\cI^{\circ},$ the interior of $\cI,$ is of dimension two and
contains no non\hyph trivial $\ell^2$ solutions.
\end{thm}
\begin{proof}
By inclusion of $\ell^p$ spaces we have $r(I)\ollim{1}R_*\neq 0$ so
that $r^{-3}(I)\ollim{1}R^{-3}_*$ by way of the Mean Value Theorem.
Therefore,
$$\frac{\fM(I)}{r^3(I)}\ollim{1}\frac{\fM_*}{R^3_*}.$$
Then, since
\begin{equation}\frac{r^2(I)\sqrt{M(I+1)}}{r^2(I+1)\sqrt{M(I)}}\ollim{1}\eta^{\gamma/2}; \,\,\,
\frac{r^2(I)\sqrt{M(I-1)}}{r^2(I-1)\sqrt{M(I)}}\ollim{1}\eta^{-\gamma/2},\label{moverr}
\end{equation}
it is not difficult to show that
\begin{equation}\label{thegs}
G_3(I)\ollim{2}\frac{(4+\gz)\gL_*}{\eta^{\gamma/2}+\eta^{-\gamma/2}};\,\,
G_2(I)\ollim{2}-\gz\gL_*
\end{equation}
Now, the result follows by applying Proposition \ref{abprop} and
Theorem \ref{first} and a scaling argument. The remaining claim
follows as above, mutatis-mutandis with (\ref{thegs}) in
$O_{\ell^{1}}$ mode.
\end{proof}
We remark: The spectrum depends on the partition $\{
r(I)|I=1,2,\hdots\};$ and, if the convergence of (\ref{thegs}) is in
$O_{\ell^1}$ mode, the solutions $\vec{X}_{\gl}$ to $A\vec{X}=\gl
\vec{X}$ for $ \gl \in $ $\cI^{\circ}$ are complex linear
combinations of vectors $\vec{Y}_{\pm}$ (Jost solutions \cite{s2})
satisfying \begin{equation} \vec{Y}_{\pm}(I)=e^{\pm i \theta
I}(1+o(1))\label{jost}\end{equation}
 (as $I\rightarrow \infty$) for
$\theta= \arccos\left[(\gl+\gz\gL_*)/(2\gL_*\kappa) \right].$
\section{Some Non\hyph polytropic Cases}\label{sc4}
We find natural examples where the assignments $\vec{\fM}$ and
$\vec{r}$ determine $\vec{\rho}$ by mass conservation. Here, and
until otherwise specified, we will suppose that  HSE holds exactly
at the surface $r=R_*$. However, we do not at this point assign an
equation of state nor any explicit dependence of $P$ on $\rho.$ In
the general case we suppose it plausible that $\Gamma P$ depends on
$\rho$ and other parameters which, in turn, may vary with $r$ as
well (see \cite{kw} and sections 4.2b and 8.7 of \cite{co}). We will
denote $\vec{\fD}:$ $\fD(I)$ $=(\Gamma P\rho)(I)$ which we will call
a {\it pressure\hyph density distribution}; and, we will call a
correspondence between the assignments $\vec{\fM},$ $\vec{r}$, such
as (\ref{massdist}) a {\it mass distribution}. We will also say that
$\vec{\fD}$ is {\it admissible} if (\ref{sys}) holds and if
(\ref{pressure}) holds but in the limit as $I\rightarrow$ $\infty$,
each in the discrete sense. We state
\begin{prop}\label{cor}
Let $\vec{\fD}$ be an admissible pressure-density distribution such
that
\begin{equation}\label{propeqn} \frac{\Gamma(I) \rho(I)}{M(I)}\ollim{1} \frac{1+\gz/4}{\pi R^3_*\cdot(1+\eta^{-\gamma})}
\end{equation}
Then, $\gs_{ac}(A)=\cI$ for $\cI$ as in (\ref{spect}). Moreover,
(\ref{mainprob}) has no non-trivial bounded solutions $\vec{\gd r}$
for $-\go^2\in \cI^{\circ}\bigcap (-\infty,0]$ if $\Gamma$ satisfies
$$\Gamma(I) = c\eta^I $$ (as $I\rightarrow \infty$) for any positive constant $c$.
\end{prop}
\noindent Some such EOS may be of the form $P=\gt\rho(r) +l(r)$ for
some constant $\gt$ and function $l(r)$ tending sufficiently rapidly
to 0 as $r\rightarrow R^-_*.$
\begin{proof}
The hypothesis (\ref{propeqn}) assures that
$$\left(\frac{P}{M}\right)(I)\ollim{1}\frac{G \fM_*}{4\pi R^4_*}=\frac{\gL_*}{4\pi R_*}$$
and that (\ref{GPoverM2}) holds in the desired mode. The result then
follows from Proposition \ref{abprop} and Theorem \ref{thm2}, with
choice of $\gz$ determined by $c,$ since the asymptotic estimates of
the $LHS$ of (\ref{propeqn}) lead to $4+\gz =c \cdot K$ for
\begin{equation}
K\df (1+\eta^{-\gamma})/(\eta^{-1}-1).\label{K}\end{equation}
\end{proof}
 For later reference we make the following
 \begin{remark} \label{firstremark} For the choice of $\vec{\fD}$ in Proposition \ref{cor} the elements $G_1(I), G_2(I)$ equal their respective limits as in (\ref{thegs}) up to order
 $O(\eta^I)$ (as $I\rightarrow\infty)$ so that
 $A_0\df A-\gk\Lambda_*J_0 + \gz\Lambda_*$ satisfies, more precisely, $\left[A_0\right]_{j,k} =O(\eta^j)$ as $j\rightarrow \infty$ for $|j-k|\le 1.$
 \end{remark}
It turns out that we may indeed choose $\vec{\fD}$ so that condition
(\ref{req}) is satisfied, for suitable $\vec{M}$ and $\vec{r},$ as
we demonstrate in
\begin{thm}\label{ex1}
Let $\vec{M},\vec{r}$ determine a mass distribution as in
(\ref{massdist}) for $\gamma$ $>1.$ Then, there is an admissible
pressure-density distribution $\vec{\fD}$ by which the results of
Proposition \ref{cor} hold for some interval $\cI$.
\end{thm}
\begin{proof}
We find
\begin{equation}
\frac{r(I+1)-r(I)}{M(I+1)}=\frac{R_*\eta^{I}(1-\eta)}{\fM_*\cdot(1-\eta^{\gamma})\eta^{\gamma
(I+1)}}.\notag\end{equation}

Since $r^2(I)\ollim{1} R^2_*$ we have from (\ref{sys}) and
L'H\^opital's Rule that $\rho(I)=O(\eta^{(\gamma-1)I})$ (as
$I\rightarrow \infty$). We may choose $\Gamma$ so that for given
$\gz>-4$
\begin{align}\Gamma(I)&=\frac{1+\gz/4}{\pi (1+\eta^{-\gamma})}\frac{M(I)}{R^3_*\rho(I)}(1 + O(\eta^{I}))\label{nonpolex}
\,\,\,(\text{\rm as  }\,I\rightarrow \infty)
\end{align}
 to apply Proposition \ref{cor}.
\end{proof}

Example (\ref{nonpolex}) satisfies $\Gamma(I)= O(\eta^I)$ (as
$I\rightarrow \infty$) and thus vanishes to order $O(R_*-r)$ as $r
\rightarrow R_*^-$. Physical models associated with such
non-polytropic cases may well incorporate certain gas and chemical
(indeed isentropic) properties: Some related discussion can be found
in Appendix B.3 of \cite{svh}, Section 8.9 of \cite{co} and Section
58 of \cite{lw}.

We will construct an admissible pressure\hyph density distribution
that results in an infinite number of eigenvalues.
 Our construction follows one by \cite{ds} whereby $pp$ spectra result in a special case satisfying $G_1(I)\equiv 1$ and $G_2(I)\ollim{2} 0$ (as $I\rightarrow \infty$). We outline the results of their construction below as
 \begin{prop}\label{ds} (Previously Known.)
One can construct the diagonal elements $G_2(I)$ of $A$ along with a
sequence of vectors $\vec{X}_m$ such that the following hold for
$\mathfrak{G}$ $\df A-J_0$,:
\begin{itemize}
\item[i)] $G_2(I)= I^{-\ga}$ for $I\in$ $\bigcup_mB_m$ for some $1/2<\ga<1,$ fixed,
 where the sets $B_m$ are bounded, disjoint, discrete intervals with distance greater than 2 between each other with
    $\min B_m$ $>Km^{p+1}$ for some positive constant $K$ uniformly in $m$;
\item[ii)] the $\vec{X}_m$'s satisfy $||\vec{X}_m||_{\ell^{\infty}}=1$ with $\supp\vec{X}_m$ $\subsetneq$ $B_m$
whose components vanish at the endpoints;
\item[iii)] for the standard inner product on $\ell^2(\mathbb{Z}^+),$ the vectors $\vec{X}_m$ are mutually orthogonal and
    $$ m^{-p}\lesssim \sum_I\left|\vec{X}_m(I+1)-\vec{X}_m(I)\right|^2\lesssim 1,$$
    so that
$-\left< \vec{X}_m, (J_0+2)\vec{X}_m\right>\gtrsim m^{-p},$ for some
$p$ satisfying $1/3<p<\ga(p+1)/2;$
\item[iv)] $\left<\vec{X}_m,{\mathfrak G}\vec{X}_m\right>\gtrsim m^{p-\ga(p+1)};$
\item[v)]$\left<\vec{X}_m,\vec{X}_m\right> \asymp m^{p};$
\item[vi)] and, $\gs_{ess}(A) =[-2,2]$.
\end{itemize}
\end{prop}
Several remarks are in order: Here $f(x)\lesssim g(x)$ means
$f(x)=O(g(x))$ as $(x\rightarrow \infty)$ and $f(x)\asymp g(x)$
means that
 $f(x)\lesssim g(x)$ and $g(x)\lesssim f(x)$ both hold; thus,
the implied bounds above are uniform in $m$. There results, by
variational arguments (Theorem 9.2 \cite{ds}), an infinite sequence
of eigenvalues $\{E_k\},$ such that $|E_k|$ $\ollim{3/2}2$, but the
convergence is not in $O_{\ell^{1}}$. Indeed, given any $q<3/2$,
$\ga$ and $p$ may be chosen so that the sequence $\{E_k\}$ diverges
in $O_{\ell^{q}}$ mode. Such may also be constructed so that
$\{E_k\}$ clusters about both $\pm 2$, but negative eigenvalues
$E_k$ better serve for our eigenfrequencies $\go=$ $\sqrt{-E_k}$.

We find that similar results arise in our problem (\ref{mainprob}).
We start with an admissible pressure-density distribution
$\vec{\mathfrak{D}}$ as above but with
\begin{equation}
\Gamma(I) = (c - b\cdot {\mathfrak
G}(I))\eta^I\label{ppexample}\end{equation}
 for ${\mathfrak G}$ as in Proposition \ref{ds} and for constants
 $1/2<\ga<1$ and $c,b>0:$
 Recalling the various parameters as in (\ref{gLkappa}) and
 (\ref{K}), we state
\begin{thm}\label{ppspectra} There is an admissible pressure-density distribution resulting in an operator
$A$ as in (\ref{mtx}) for which $\gs_{pp}(A)$
contains an infinite sequence $E_k:k=1,2,\hdots$ tending to
$-(\gz+2\gk)\gL_*$, an endpoint of $\gs_{ess}(A)=\cI$, from values
less. Moreover, for each $k$ there is an eigenvector associated to
$E_k$ for which the corresponding perturbation $\vec{\delta r}$ is
of class $\ell^{\infty}(\BZ^+)$.
\end{thm}
\begin{proof}
It will be convenient to prove the case for $\gL_*\gk=1$, choosing
$\Gamma$ as in (\ref{ppexample}), and letting $c=\frac{4}{K}$
 and $b=\frac{1}{\gL_* K}$
from which we obtain $\gz=0,$
\begin{align}G_3(I)&=1 - \frac{\gk}{4} \fG(I)+O(\eta^I)\notag\,\,\text{\rm
and}\\
G_2(I)&=\frac{1}{4\gL_*}\left(\eta^{\gamma/2}\fG(I)+\eta^{-\gamma/2}\fG(I-1)\right)
+O(\eta^I).\end{align} Since $\fG$ is constant on $B_m$, we have
from i) of Proposition \ref{ds} that
$$\left<\vec{X}_m,\fG(I-1)\vec{X}_m\right>=\left<\vec{X}_m,\fG(I)\vec{X}_m\right>.$$
If we denote $G_2\df$ diag$(G_2(1),G_2(2), \hdots)$, we find
$$\left<\vec{X}_m, G_2\vec{X}_m\right> = \left< \vec{X}_m, \frac{1}{\gk \gL_*}\fG \vec{X}_m \right>=1\cdot\left< \vec{X}_m,\fG \vec{X}_m \right> ,$$ modulo $O(m^p\eta^m).$
Let us set $\tilde{A}\df $ $A-J_0-G_2$ which is also self-adjoint.
Considering Remark \ref{firstremark} we find that the elements along
the main diagonal of $\tilde{A}$ are of order $O(\eta^I)$ and those
along super- and sub-main diagonals are appraised via
$[\tilde{A}]_{I,I+1}\lesssim I^{-\ga}+\eta^I.$ With $\ga>1/2$ we
conclude from the Weyl Invariance theorem that $\gs_{ess}(A)$
$=[-2,2].$
 Here, $$\left|\left<\vec{X}_m,\tilde{A}\vec{X}_m\right>\right| \lesssim  \cdot m^{-\ga(p+1)}\sum_I|\vec{X}_m(I+2)-\vec{X}_m(I)|^2+
 \eta^m\cdot||\vec{X}||^2$$
\begin{align}\lesssim & m^{-\ga(p+1)}\sum_I|\vec{X}_m(I+1)-\vec{X}_m(I)|^2 + \eta^m m^p
\notag
\end{align}
 Hence, we find that there are positive constants
$\kappa_1$ and $\kappa_2$ so that
$$\left|\left<\vec{X}_m,\tilde{A}\vec{X}_m\right>\right|< \kappa_1m^{-\ga(p+1)}+\kappa_2\eta^m m^p
$$ holds uniformly for all sufficiently large $m.$
It follows as in Proposition \ref{ds} that
$$
-\left<\vec{X}_m,(J_0+G_2+2)\vec{X}_m\right>
> Cm^pm^{-\ga(p+1)}
$$
for some positive constant $C$.

We have only to follow a part of the cited theorem \cite{ds},
mutatis-mutandis: There are positive constants $M$ and $C_{M}$ such
that
$$-\left<\vec{X}_m,(A+2)\vec{X}_m\right> >C_{M} ||\vec{X}_m||^2m^{-\ga(p+1)}$$
$\forall m\geq M.$ Variational arguments now apply to show that
supported in each $B_m:m\geq M$ is a corresponding eigenvector with
eigenvalue $E_m<-2$ such that $|E_m+2|\geq$
$\frac{C_{M}}{2}m^{-\ga(p+1)},$ concluding the present case.

If we keep $b$ as above and vary $c,$ we may apply the above
construction but with $A$ replaced by $A+\gz\gL_*$ for an
appropriate constant $\gz$ (see Proposition \ref{cor}). Considering
Remark \ref{firstremark},  $\tilde{A}$ likewise yields elements of
order $O(\eta^I)$ along the super- and sub-main diagonal elements.
The result follows but with a shift of the spectral values by
$-\gz\gL_*$ where $\gz$ is determined by $c$. The more general case
now follows by scaling arguments.
\end{proof}

We note the following about the results above: Such a construct
shows that a small variation of $\Gamma$ can change the discrete
spectrum $\gs_{disc}(A)$. Moreover, as a sequence of eigenvalues
converges to particular value near the stellar surface, this value
does not necessarily correspond to a value in $\gs_{disc}$. We see
that ${\mathfrak G}(I)$ in (\ref{ppexample}) vanishes as
$r\rightarrow R_*^-$ to order $O([\ln(R_*-r)]^{-\ga})$ and, hence,
the entire expression for $\Gamma$ still vanishes to order
$O(R_*-r)$ as in example (\ref{nonpolex}). Yet, the additional term
$-b{\mathfrak G}$ has an demonstrable effect on the spectrum of $A$,
even for small $b>0,$ yielding oscillation states localized near
$r=R_*$.
\section{Polytropic Cases: Modified Difference Equation}\label{sc5}
 We now apply our methods to some polytropic states in the discrete setting.
 In order to apply the Jacobi spectral matrix methods we resort to another change of variables.
 Here, we allow the oscillation frequency $\go$ to vary with position, while freezing time:
 We impose $\gd r_t$ $=$ $e^{i t \go(r)}\gd r(r),$ supposing that the
linear perturbation $\gd$
 and the operation $\dert$ commute (albeit in a formal sense; cf. \cite{svh}). Such modes satisfy
 $[\partial_r^j\delta r_t|_{t=0}=\partial_r^j\delta r(r) : j=0,1,2,$
 leaving (\ref{mainprob}) still valid.
Then, with
 \begin{equation}
 \cW\df \text{\rm diag} (\go(1),\go(2),\hdots,\go(I),\hdots)\label{alphas},
 \end{equation}
we obtain an infinite matrix equation of the form
\begin{equation}
-\cW^{2}\vec{X}=A\vec{X}\label{modified}
\end{equation}
As we develop below, with an appropriate choice of $\cW,$ depending
on $\eta$ and $\gl$, (\ref{modified}) leads to the form
(\ref{mainprob}) but with $LHS$ $=-\go^2(I)\gd r(I).$ Regarding the
spatial dependence of $\go$, we will call the values $\go(I)$ {\it
local frequencies}.

We now motivate our formulation (\ref{modified}) by way of a certain
modification of the Jacobi matrix: We will study tri-diagonal
matrices $A$ given by
$$G_3(I)=\gm_I\eta^{-I}; \,\,G_2(I)=\theta-\beta_I\eta^{-I}$$
(self-adjoint) where
\begin{equation}
\gm_j\ollim{1} \gm;\,\, \gb_j\ollim{1}\gb\label{mubeta}
\end{equation}
for some positive numbers $\eta,\theta,\gb, \gm$ such that $\eta<1$.

We now reformulate our matrix equation (\ref{mtx}): First, we
consider square, $n\times n$ matrices (for $n>3$, say) of the form
\begin{equation}A(x,y)=\left(\begin{matrix}
a_{1}&c_1 x&0&0&\hdots&0\\
b_1 y& a_2 & c_2x^2& 0 &\hdots &0\\
0 & b_2 y^2 & a_3& c_3 x^3&\hdots&0\\
\vdots &\hdots &\ddots& \ddots& \ddots &\vdots\\
0 &\hdots & 0& b_{n-2}y^{n-2}& a_{n-1} &c_{n-1} x^{n-1}\\
0&0&\hdots&0 & b_{n-1} y^{n-1}&a_n\\
\end{matrix}\right)\label{newmtx}\end{equation}

\begin{lem}
Let $A$ be a square, tri-diagonal matrix of the form (\ref{newmtx}).
Then, there is a diagonal matrix $$D(x,y)= \text{\rm
diag}(d_1(x,y),d_2(x,y), \hdots,d_n(x,y))$$ for which the following
hold:
\begin{equation}\label{B}
B(x,y)=D(y,x)A(x,y)D(x,y)
\end{equation}
is a tri-diagonal matrix whose elements $[B]_{i,j}$ satisfy
\begin{equation}
[B]_{i,j}=\begin{cases} [A(1,1)]_{i,j} & \mbox{for } i\neq j\\
\frac{a_j}{(xy)^{\lfloor\frac{j}{2}\rfloor}} & \mbox{for } i= j
\end{cases}
\label{r};\end{equation} and, the $d_j(x,y)$'s are rational
monomials in the variables $x$ and $y$.
\end{lem}
Here $\lfloor\cdot\rfloor$ denotes the integer part of the argument.
\begin{proof}
Our choice for diagonal elements of $D(x,y)$ are as follows:
$$d_j(x,y)=\begin{cases} 1 &\mbox{for  } j=1\\
\prod_{k=1}^m\left(\frac{y^{2k-2}}{x^{2k-1}}\right)& \mbox{for  } j=2m , \text{\rm even}\\
\prod_{k=1}^m\left(\frac{y^{2k-1}}{x^{2k}}\right)& \mbox{for  }
j=2m+1 , \text{\rm odd}
\end{cases}
$$
We verify our results by considering row and column operations
applied to $A(x,y)$ according to (\ref{B}): For $j=2m\geq 2$, even,
we find
$$d_{j}(x,y)d_{j+1}(y,x)=\prod_{k=1}^m\left(\frac{y^{2k-2}}{x^{2k-1}}\right)\left(\frac{x^{2k-1}}{y^{2k}}\right)
=\prod_{k=1}^{m}\frac{1}{y^{2}}=\frac{1}{y^{2m}}=\frac{1}{y^j}$$
$$d_{j}(x,y)d_{j-1}(y,x)=\left(\frac{y^{2m-2}}{x^{2m-1}}\right)
\prod_{k=1}^{m-1}\left(\frac{y^{2k-2}}{x^{2k-1}}\right)\left(\frac{x^{2k-1}}{y^{2k}}\right)
=\frac{1}{x^{j-1}}
$$ and
$$d_j(x,y)d_j(x,y)=\prod_{k=1}^m\left(\frac{y^{2k-2}}{x^{2k-1}}\right)\left(\frac{x^{2k-2}}{y^{2k-1}}\right)
=\frac{1}{(xy)^m}=\frac{1}{(xy)^{j/2}}
$$
For $j=2m+1$, odd,
$$d_j(x,y)d_j(x,y)=\prod_{k=1}^m\left(\frac{y^{2k-1}}{x^{2k}}\right)\left(\frac{x^{2k-1}}{y^{2k}}\right)
=\frac{1}{(xy)^{m}}=\frac{1}{(xy)^{(j-1)/2}}$$ Hence, we
obtain\begin{align}
 [B]_{j+1,j}&=[A]_{j+1,j}d_{j}(x,y)d_{j+1}(y,x) = b_j \notag\\
 [B]_{j,j+1}&=[A]_{j,j+1}d_{j}(y,x)d_{j+1}(x,y) = c_j\notag\\
 [B]_{j,j}& = [A]_{j,j}d_{j}(x,y)d_j(x,y)=\frac{a_j}{(xy)^{\lfloor j/2\rfloor}}\notag.\end{align}
\end{proof}
We note that the results do not depend on the dimension $n$ and thus
hold for the infinite case with elements $d_j(x,y),[B]_{j,j}:j\in
\BN$ defined as above.

Now, setting $\ga(I)={\lfloor I/2\rfloor},
\gl_I=-\gl+\beta_I\eta^{2\ga(I)-I}$ (for parameter $\gl\leq 0$),
$\go(I)=\sqrt{\gl_I}\eta^{-\ga(I)}$, and $x=y=\eta^{-1}$ in
(\ref{newmtx}), we use Lemma \ref{B} now to convert equation
(\ref{modified}) into a Jacobi-matrix eigenvalue problem of the form
\begin{equation}\gl \vec{Y}=\cT\vec{Y}\label{neweqn},\end{equation}
for a certain tri-diagonal $\cT.$ With $\cW= H^{-1}\cL$ we denote
the following:
$$\cL\df \text{\rm
diag}\left(\sqrt{\gl_1},\sqrt{\gl_2},\hdots,
\sqrt{\gl_j}\hdots\right);
$$
$$H\df \text{\rm diag}\left(\eta^{\ga(1)},\eta^{\ga(2)},\hdots, \eta^{\ga(j)}\hdots\right)$$
so that $-\cW^2X=H^{-2}(\cL^2 X).$ We then set $\cA\df H A H$
$\df\cT-\cD$ for $H$ as in (\ref{modified}) where
\begin{equation}\cT\df
\left(\begin{matrix} \gth&\gm_1&0&0&0&0&\hdots\\
\gm_1&\gth\eta&\gm_2&0&0&0&\hdots\\
0&\ddots&\ddots&\ddots&0&0&\ddots\\
\vdots&0&\gm_{j}&\gth\eta^{2\lfloor j/2\rfloor}&\gm_{j+1}&0&\ddots\\
\vdots&\vdots&\ddots&\ddots&\ddots&\ddots&\ddots
\end{matrix}\right)\label{matrix}
\end{equation} and
\begin{align}
\cD\df\text{diag}\left(\gb_1/\eta,\gb_2,\gb_3/\eta,\gb_4,\hdots,\gb_j\eta^{2\ga(j)
-j},\hdots\right).\label{diagonalmtx}\end{align} Since $\vec{Y}=
H^{-1}\vec{X},$ equation (\ref{modified}) leads to
$$H^{-1}(\gl Y -\cD Y)=H^{-1}(\cT-\cD)Y,$$ yielding (\ref{neweqn}).

We are ready to state
\begin{prop}\label{mubetaprop} The absolutely continuous spectrum of the operator $\cT$ as in (\ref{neweqn}) is
$[-2\gm/\eta,2\gm/\eta]$ for $\gm$ as in (\ref{mubeta}), the
interior of which contains no pure-point spectra.
\end{prop}
\begin{proof}
This is a consequence of part iii) of Theorem \ref{first} applied to
(\ref{neweqn}).
\end{proof}

It is worthwhile to point out that the nearly periodic nature of the
elements of $\cD$ leads to the presence of spectral gaps under an
appropriate formulation. If we replace $\vec{\go}$ by
$\vec{\go}_{\sharp}$ where $\go_{\sharp}(I)=$
$\sqrt{-\gl}\,\eta^{-\ga(I)}$ and replace $\cW^2$ by $\gl H^{-2}$,
equation (\ref{neweqn}) can be recasted as $\gl\vec{Y}=\cA\vec{Y}$ :
We thereby convert equation (\ref{mainprob}) into a weighted
eigenvalue problem, now with $LHS$= $-\gl\eta^{-2\ga(I)}$ $\gd
r(I).$ We further note that $\cA$ is a perturbation of a Jacobi
matrix with periodic coefficients and, in our search for negative
eigenvalues, make the following
\begin{remark}\label{newremark}
The absolutely continuous spectrum of $\cB\df-\cA=-\cT+\cB$ is given
by
$$\gs_{ac}(\cB)=[E^-,E_1]\bigcup[E_2,E^+]$$
where $E^{\pm}=\frac{\gb(1+\eta)\pm\sqrt{16\eta^2\gm^2
+\gb^2(1-\eta)^2}}{2\eta}$ (resp.), $E_1=\gb,$ and
$E_2=\frac{\gb}{\eta}.$ Moreover, the pure-point spectra is
restricted to the spectral gaps (ie. $\gs_{pp}(\cB)\subset$
$\BR\setminus \gs_{ac}(\cB)$) and $\gs_{ess}(\cB)$ is purely
absolutely continuous.
\end{remark}
\noindent See Theorem 7.11 along with equations (7.71) of \cite{t};
or, see Chapter 5 of \cite{s2}. Although the $\go_{\sharp}(I)$ of
Remark \ref{newremark} do not depend on the $\beta_I,$ the
associated spectral sets still depend on $\gb$. Moreover, we see
that the essential spectrum still depends on the choice of partition
(determined by $\eta$) in either case. Thus, out of convenience, we
continue with the formulation (\ref{neweqn}).

Recall that for a polytropic state there are positive constants $K,
\Gamma$ so that $P=K\cdot\left(\rho(r)\right)^{\Gamma}.$ With mass
distributions given by (\ref{massdist}) for $\gamma>1$, we now
consider distributions $\vec{\fD}$ that are admissible, thereby
approximating polytropes near the boundary $r=R_*$: That is, for a
polytopic state, we suppose
\begin{align}\left(\frac{P}{M}\right)(I)&=G\fM_*/(4\pi R^4_*)+O(\eta^{(1+\epsilon)I}) \label{polycond1}
\\
\rho(I)&=\frac{M(I)}{4\pi
r^2(I)(r(I)-r(I-1))}(1+O(\eta^{(1+\epsilon)I}))\label{polycond2}
\end{align}
(as $I\rightarrow \infty$) for some constant $\epsilon>0$. Under
these conditions such a $\vec{\fD}$ will be called an {\em almost
polytrope}. Using a version of the Mean Value Theorem, it is not
difficult to show, for instance, that (\ref{polycond1}) and
(\ref{polycond2}) both hold for our mass distributions if
$\rho(r)\propto(R_*-r)^{\gamma-1}$ with
$\gamma=\frac{\Gamma}{\Gamma-1}>1$ (taking $\epsilon=\gamma-1>0$).

For ease of exposition, we start with an almost polytrope
$\vec{\fD}$ where $\gamma\geq 2.$ Noting $(M/\rho)(I)$
$=O(\eta^{I})$ as $I\rightarrow\infty$, we introduce\begin{equation}
\gm_I\df\eta^{I}G_{3}(I) ; \,\,\gb_I\df \eta^I (G_2(I)-4\gL_*).
\label{mu}\end{equation} Then, since $\frac{G_1(I)}{G_3(I)}$ $=\eta
+O(\eta^{ 2I})$ we obtain from (\ref{Gs}) and (\ref{moverr}), along
with (\ref{polycond1}) and (\ref{polycond2}), that
\begin{align}
G_3(I)&=\gL_*\Gamma\cdot\frac{\eta^{1-\gamma/2}}{1-\eta}\eta^{-I}(1+O(\eta^{2I}))\notag\\
G_2(I)&=4\gL_*-\gL_*\Gamma\cdot\frac{\eta+\eta^{1-\gamma}}{1 -
\eta}\eta^{-I}(1+O(\eta^{2I})) \notag
\end{align}
(as $I\rightarrow\infty$) and find that $\gm_I$ and $\gb_I$ thereby
satisfy
\begin{equation}
\label{bmest}\gm_I-\frac{\gL_*\Gamma\cdot\eta^{1-\gamma/2}}{(1-\eta)}\lesssim
\eta^{I};\,\,\, \gb_I-
\gL_*\Gamma\cdot\left(\frac{\eta+\eta^{1-\gamma}}{1-\eta}\right)
\lesssim \eta^{I}.\end{equation} With $\beta\df\lim_{I\rightarrow
\infty}\beta_I$ and $\gm\df\lim_{I\rightarrow \infty}\gm_I$, the
estimates of $\gm_I,$ and $\gb_I$ result in
\begin{thm}\label{mubetathm}
  For an almost polytrope $\vec{\fD}$ with $1<\Gamma \leq 2$
(equivalently, $\gamma\geq 2$) let $A$ and $\cA$ be defined by
$G_3(I),$ $\gm_I,$ $G_2(I),$ and $\gb_I$ as in (\ref{mu}). Then, the
result of Proposition \ref{mubetaprop} holds for $\cA$. Furthermore,
if $\gl$ $\in$ $\gs_{ac}(\cA)\bigcap(-\infty,0],$ then any $\vec{\gd
r}$ associated with a non-trivial solution vector $\vec{X}$ of
(\ref{modified}) satisfies $\vec{\gd r} \notin \ell^{\infty}
(\BN^+)$; and, the local frequencies $\go(I)$ satisfy
$\lim_{I\rightarrow \infty }\go(I)$ $=+\infty$.
\end{thm}
\begin{proof}
We apply the results of Proposition \ref{mubetaprop} and Theorem
\ref{thm2} with $\theta=4\gL_*$ and $\gb_j,$ $\gm_j$ for
$j=1,2,\hdots$ as in (\ref{mubeta}). We note that the estimates
(\ref{bmest}) lead to $\vec{Y}$ $\notin$ $\ell^2$ for any such $\gl$
$\in\gs_{ac}(\cA)$: Moreover, estimates (\ref{jost}) apply for a
basis $\vec{Y}_{\pm},$ of the solution space. We therefore have
$Y(I_k)$ $\asymp$ $1$ as $k\rightarrow \infty$ for a subsequence of
indices $I_k$. Then, since $\vec{X}=H\vec{Y}$ and $\gamma>1,$ we
have that
\begin{equation}\label{newdr}\gd r(I_k)
=\frac{\eta^{\ga(I_k)}Y(I_k)}{\sqrt{M(I_k)}} \asymp \eta^{(
\ga(I_k)-\gamma I_k/2)}\gtrsim
\eta^{-I_k(\gamma-1)/2},\end{equation} (as $k\rightarrow$ $\infty$)
resulting in $\limsup_{I\rightarrow \infty}|\gd r(I)|=\infty$.

Finally, the unboundedness of $\go(I)$ follows since
\begin{equation}\label{newlamda}\liminf_{I\rightarrow \infty}\gl_{I}=-\gl+\gb\eta>0\end{equation}
and, hence, $\go(I)\gtrsim \eta^{-I/2}$ (as $I\rightarrow \infty).$
\end{proof}

Following the above construction, we treat an alternative type of
polytropic distribution where density vanishes at the surface but
$\gamma$ is nearly $1$ (whereby $\rho$ is nearly constant slightly
away from surface): We suppose only condition (\ref{polycond2}) and
not necessarily (\ref{polycond1}), thereby imposing mass
conservation but not excluding the almost polytrope. In this case we
set
\begin{equation}\left(\frac{P\rho}{M^2}\right)(I)=C_*\cdot\eta^{I((\gamma-1)(\Gamma-1)-2)}(1+O(\eta^I))
\label{polycond3}\end{equation} for constants $C_*>0$ and $\gamma,
\Gamma>1$ with $0<(\gamma-1 )(\Gamma-1)\leq 1.$ We now denote
$\gn\df  \eta^{2-(\gamma-1)(\Gamma-1)}$ and set
\begin{equation}\label{newH}H=\text{\rm diag}\left(\gn^{\lfloor 1/2\rfloor},\gn^{\lfloor 2/2\rfloor},\gn^{\lfloor3/2\rfloor}\hdots,\gn^{\lfloor I/2\rfloor},\hdots\right)\end{equation}
so as to recast our eigenvalue problem into the form
(\ref{modified}) with $G_2,G_3$ so that
\begin{align}
\tilde{\fg}(I)&\df\gn^{I}G_{3}(I)=
16\pi^2R^4_*\eta^{-\gamma/2}C_*\cdot\Gamma +O(\eta^{I})
\notag\\
\tilde{\fh}(I)&\df \gn^{I}G_2(I)= \gn^{I}
4\gL_*-16\pi^2R^4_*C_*\Gamma\cdot\left(1+\eta^{-\gamma}\gn\right))+O(\gn^{I}\eta^I).
\notag\end{align}

For convenience we now set\begin{align}
\tilde{\gb}&=-\lim_{j\rightarrow\infty}\tilde{\fh}(j)=
16\pi^2R^4_*C_*\Gamma\cdot\left(1+\eta^{-\gamma}\gn\right))\label{newbeta}\\
\tilde{\gm}&
=\lim_{j\rightarrow\infty}\tilde{\fg}(j)=16\pi^2R^4_*\eta^{-\gamma/2}C_*\cdot\Gamma
\label{newmu}\end{align} to state
\begin{thm}
 Suppose the mass distribution is as in (\ref{massdist}) with $1<\gamma\leq\frac{\Gamma}{\Gamma-1}$ with a polytropic $\vec{\fD}$ satisfying (\ref{polycond2}) and (\ref{polycond3}) for some constant $\Gamma >1$. Then the conclusions of Theorem \ref{mubetathm} likewise hold, but for
 \begin{equation}\gs_{ac}(\cA)=[-2\tilde{\gm}/\gn,2\tilde{\gm}/\gn]\notag
 \end{equation}
  with $\tilde{\gm}$ as in (\ref{newmu}).
\end{thm}
\begin{proof} For our choices of $\gamma,\Gamma$ we have $0<\gn<1$ so that we may apply the
decomposition $-HAH=\cT -\cD$ as in (\ref{matrix}) and
(\ref{diagonalmtx}), and likewise apply the analysis of Proposition
\ref{mubetaprop} and Theorem \ref{thm2}, but with $\theta=0,$ $\eta$
replaced by $\gn,$ and obvious substitutions for parameters
$\beta_j,\beta,\gm_j,\gm$ via (\ref{newmu}) and (\ref{newbeta}). In
particular, we have $\gn\leq\eta$ so that results concerning
$\vec{\gd r}$ and $\vec{\go}$ from the redefinition (\ref{newH})  of
$H$ follow as in (\ref{newdr}) and (\ref{newlamda}).
\end{proof}

We end this section by noting the following about our polytropic and
non-polytropic cases, as well as the almost polytrope: The essential
spectra can admit intervals of arbitrary length as we are allowed to
choose arbitrarily small $\eta$ (or $\gn$) in choosing our
partitions via $\vec{M}$ and $\vec{r}$. This suggests perhaps that
unbounded intervals of such spectra may be present in some such
cases as values of our discrete parameters $m$ and $r$ pass to a
continuous interval. This expectation is indeed borne out in the
remaining sections of this article.
\section{Differential Equation Model}\label{sc6}
In the continuous case we study the perturbed mass distribution $\gd
r(r)$ by way of equation (8.6) of \cite{co} (see also \cite{acdw})
\begin{equation} \label{ode}-\fder{r}\left(\Gamma P
r^4
\fder{r}\xi\right)-\left(r^3\fder{r}\left[\left(3\Gamma-4\right)P\right]\right)\xi=
\gs^2\gr r^4 \xi
\end{equation}
where $\xi(r)=\gd r(r)/r$. The analysis that we apply requires no
special boundary conditions, but we will check our models against a
so-called {\it regularity condition} given by \begin{equation}
\label{bdyc} \lim_{r\rightarrow
R^-_*}(P\Gamma)(r)(3\xi(r)+r\xi^{\prime}(r))=0
\end{equation} at the (finite) boundary $r=R_*>0$
\cite{bs,lw}. Regarding HSE, we note that $\xi$ are assumed to be
perturbations about 0 of
\begin{equation} \fder{r} P(r)+\frac{G\cdot \rho(r)}{r^2}.
\label{hsecont}\end{equation}

We (re-)introduce notation in accord with some of our references: We
set $x=r,$ $y=\xi,$ $p= \Gamma P x^4,$ $\gl =\gs^2$ (switching sign
convention), $q=-x^3\fder{x}\left[\left(3\Gamma-4\right)P\right],$
and $w= \gr x^4,$ whereby equation (\ref{ode}) takes the SL form
\begin{equation}
-\left(p(x)y^{\prime}(x)\right)^{\prime}+q(x)y(x)=\gl
w(x)y(x)\label{mainode}
\end{equation}
Here we take $\prime$ to mean the full derivative with respect to
the independent variable as indicated.

By way of the Liouville transform \cite{br} we may further convert
(\ref{ode}) to the canonical form
\begin{equation}-Y^{\prime\prime}(X)+Q(X)Y(X)=\gl
Y(X)\label{othermain}\end{equation} where, for some fixed positive
constants $R_{\gd}$ $<R_*$ $<\infty,$ \begin{align} X(x)\df&
\int_{R_{\gd}}^x \sqrt{w(t)/p(t)}\,dt\label{q}\,\,\text{\rm and}\\
Y(X)(x)\df&\left(p(x)w(x)\right)^{1/4}y(x)\notag\\
Q(X)(x)\df&\frac{q(x)}{w(x)}-\left({\frac{p(x)}{(w(x))^3}}\right)^{1/4}
\left(\left(\frac{p(x)}{w(x)}\right)^{1/2}\left((p(x)w(x))^{1/4}\right)^{\prime}\right)^\prime
\notag\end{align} with the standing assumptions that $q,p,w,1/w,$
and $\sqrt{w/p}$ are continuous on $[R_{\gd},R_*)$.

We introduce notation and terminology here to accommodate more
precise estimates needed in the following sections: It will be
convenient to denote $q_0(x)\df Q(X)(x),$  $q_1(x)\df
\frac{q(x)}{w(x)}$ and $q_2(x)\df q_0(x)-q_1(x).$ From the
invertibility of (\ref{q}) we may, with obvious notation, likewise
decompose $Q_0(X)=$ $Q_1(X)+$ $Q_2(X).$ The phrase "near $R_*$"
("near $\infty$") regarding an estimate or bound will mean such will
hold on some interval of the form $[x_0,R_*)$ (resp.
$[x_0,\infty)$).

Our strategy in the next sections will be as follows: We apply
results of subordinancy theory where general boundedness of
(non-trivial) solutions yields non-$L^2$ behavior
(counter-intuitively); in turn, the later property yields unbounded
oscillation $\xi.$ Along these lines, we make frequent use of the
results discussed below.
\begin{thm}\label{B1} (Already Known)
Suppose $Q(X)=V_1(X)+V_2(X)$ defined on $(0,\infty)$ where $V_1$
$\in$ $L^1$ and $V_2$ $\in$ $C^1$ where $V_2^{\prime}$ $\in$ $L^1$
with $\lim_{X\rightarrow \infty}V_2(X)=0$. Then, given $\gl>0,$ to
every solution $Y$ of (\ref{othermain}) there are constants $\ga$
and $\gb$ so that
\begin{align}
Y(X)&=\ga u_+(X)+\gb u_{-}(X) + o(1)\notag\\
Y^{\prime}(X)&=i\ga u_+(X)-i\gb u_-(X)+o(1)\notag
\end{align} (as $X\rightarrow \infty$) where
$ u_{\pm}(X)=\exp(\pm i\int_{X_0}^{X}\sqrt{\gl^2-V_2(t)}\,dt) $
(resp.) for some sufficiently large, fixed $X_0>0$.
\end{thm}
This is Theorem B.1 \cite{s} (also see Theorem 6 \cite{st}). We then
find unbounded oscillations for a star of finite mass $\fM_*$ and
radius $R_*$ ($0<\fM_*,R_*<\infty$) as we apply
\begin{prop}\label{unbddosc}
If a solution $Y$ of (\ref{othermain}) is not square integrable on
the domain $[0,X(R_*))$, then the corresponding $\gd r(r)$ $=\xi(r)
r$ satisfies $$\limsup_{r\rightarrow R^-_*} |\gd r(r)|=\infty.$$
\end{prop}
\begin{proof}
We compute
\begin{align}
\int_0^{X(x)}|Y(t)|^2dt&=\int_{R_{\gd}}^x \xi^2(s) s^4\rho(s)\,ds\label{intest}\\
 &=\int_{R_{\gd}}^xs^2\rho(s)(\gd r(s))^2 ds\notag\\
&\leq \sup_{s\in(R_{\gd},x)}(\gd
r(s))^2\cdot\int_{R_{\gd}}^x\rho(s)s^2ds. \notag
\end{align}
We have $LHS$ of (\ref{intest}) is unbounded as $x\rightarrow R_*$
while $\int_{R_{\gd}}^{x}\rho(s)s^2ds$ $\leq$ $\fM_*/(4\pi)$ so that
the desired result is clear.
\end{proof}
\section{A Polytropic Outer Shell}\label{sc7}
As in Section \ref{sc5} we consider an EOS of the form
$P=K(\rho(x))^b$ with $\rho(x)=(R_*-x)^a$ some fixed $a,b>0.$
Likewise, we have that $\Gamma$ $=b$ $>1$ is constant on an interval
$[R_{\gd},R_*)$ for some fixed, positive $R_{\gd}\df$ $R_*-\gd<R_*.$
After a change of variables we find that for $a(b-1)>1$ our
eigenvalue problem amounts to certain $L^1$ perturbations of a
simple operator whereby the existence of absolutely continuous
spectra is clear by the Kato-Rosenblum Theorem \cite{rs} (see also
Chapter 11, \cite{p}). For the case $a(b-1)>2$ we will further
elaborate on the behavior of solutions.

We introduce the following notation: Let $\fT_{\gl}$ be the set of
solutions to (\ref{othermain}) on $[0,\infty)$; let $S$ be the
defined by $S\df \{\gl>0 |Y\in \fT_{\gl} \Rightarrow
\sup_{[0,\infty)}|Y(X)|<\infty\}$ and, let $W(x)\df
\sqrt{w(x)/p(x)}.$ For the operator $\cL\df -\frac{d^2}{dX^2}+Q(X)$
on $L^2((0,\infty)),$ the following suffices for our applications to
follow:
 \begin{thm}\label{maincont}
Suppose $W$ $\notin L^1([R_{\gd},R_*),dx)$ and that $$q_0(x)\in
\cC([R_{\gd},R_*]) \bigcap L^1([R_{\gd},R_*);W(x)\,dx).$$ Then,
$\gs(\cL)$ $=\gs_{ac}(\cL).$  Indeed, $S$ $\supseteq$ $(0,\infty)$
and, moreover, for any fixed $\gl$ $>0,$ $Y^{\prime}$ is bounded on
$[0,\infty)$ for every $Y$ $\in$ $\fT_{\gl}.$
\end{thm}
\begin{proof}
The result is an immediate consequence of Theorem 1 \cite{s} and
Theorem \ref{B1} from our change of variables (\ref{q}).
\end{proof}
\noindent More precisely, the spectral measure $\fl(\gl)$ associated
with $\cL$ is absolutely continuous and is supported in
$[0,\infty),$ and $\fl$ is absolutely continuous in the sense that
$\fl(U)>0$ for any (Lebesgue) measurable $U$ $\subset$ $S$ of
positive Lebesgue measure.

For our polytropic case, $P=K\rho^b$ with $\rho=(R_*-x)^a,$ we state
\begin{prop}\label{polytrop}
Given fixed $a,\gd,K>0,$ with $0<R_{\gd}<R_*,$ and $b>1,$ the
following hold for $q_0(x)=$ $Q(X)(x):$\begin{itemize}
\item[i)]
$q_0(x)$ is smooth on $[R_{\gd},R_*)$;
\item[ii)]
$q_0(x)$ $\asymp(R_*-x)^{ab-a-2}$ near $R_*;$ and,
\item[iii)] $q_0(x)-k_{a,b}\in $ $L^1([R_{\gd},R_*),dX(x))$ if $1< ab-a\leq
2$ for some constant $k_{a,b},$ depending on $a$ and $b.$
\end{itemize}
\end{prop}
\begin{proof}
Here $p(x)=bK(R_*-x)^{ab}x^4$ and $w(x)=(R_*-x)^ax^4$ and a
computation shows that
$$ q_1(x) = K\frac{ab(-4+3 b)(R_* - x)^{ab-a-1}}{x}
$$
and
$$
q_2(x)=-KR^2_*\frac{b (R_* - x)^{ab -a  - 2}}{16 x^2}\fQ(x/R_*)
$$
where
\begin{equation}\fQ(u)=32-32(2+a b)u+(32 +4a(-1+7
b)+a^2(-1+2b+3b^2))u^2. \label{poly}\end{equation} In the case
$ab-a=2$ we have that $q_0(x)-k_{a,b}$ $\lesssim (R_*-x)$ near $R_*$
for $k_{a,b}=q_0(R_*);$ hence, $(q_0(x)-k_{a,b})W(x)$ is bounded on
$[R_{\gd},R_*).$ For $1<ab-a<2$ we set $k_{a,b}=0$ as we find that
$q_0(x)W(x)$ is absolutely integrable on $[R_{\gd},R_*).$
\end{proof}
\noindent We notice that
 $\fQ(u)=0$ has a root $u=1$ when $a=4/(3b-1),$ but we will not use this fact considering other restrictions on $a$ and $b$.

We now perform a change of variables
\begin{equation}X=\int_{R_{\gd}}^xW(t)\,dt=\int_{R_{\gd}}^x(R_*-t)^{(a-ab)/2}/\sqrt{b
K}\,dt\label{cov}\end{equation} for positive $a,b>1$ with $a(b-1)>2$
whereby $X$ takes values of $[0,\infty).$ Here $Q(X)$ is smooth and
bounded and $\lim_{X\rightarrow \infty}Q_1(X)$ $=0.$ We state
\begin{thm} For the EOS as in Proposition \ref{polytrop} with
$a$ and $b$ as in (\ref{cov}) we have the following:
\begin{itemize}
\item[i)]The result of Theorem \ref{maincont} holds;
\item[ii)] the
regularity condition (\ref{bdyc}) holds; and,
\item[iii)] for each $\gl>0$ every corresponding non-trivial $\gd r(x)$ is
unbounded.
\end{itemize}
\end{thm}
\begin{proof}
We integrate over $(R_{\gd},R_*)$ the quantity
$$Q(X)(x)dX(x) = q_0(x)W(x)dx$$
with $q_0(x)W(x)$ $\asymp$ $(R_*-x)^{\frac{(ab-a)}{2}-2}$ whereby
Theorem \ref{maincont} applies to prove i).

Addressing ii), we write
\begin{equation}\label{y}
y(x)=\frac{Y(X)(x)}{(pw)^{1/4}(x)}=\frac{Y(X)(x)}{x\left(\rho(x)bP(x)\right)^{1/4}}.\end{equation}
For any solution in $\fT_{\gl}$ we have that $Y$ and $Y^{\prime}$
are bounded and find
\begin{align}(P\Gamma)(x)y(x)&\lesssim (R_*-x)^{\frac{3ab-a}{4}}\lesssim (R_*-x)^{(a+3)/2}\notag\\
(P\Gamma)(x)y^{\prime}(x)&\lesssim
(R_*-x)^{\frac{3ab-a}{4}-1}\lesssim (R_*-x)^{(a+1)/2}\notag \notag
\end{align}
near $R_*$ so that (\ref{bdyc}) holds.

Finally, by Theorem \ref{B1}, $Y(X)$ $\asymp$ $1$ near $\infty$ so
that result iii) follows from Proposition \ref{unbddosc}.
\end{proof}
\noindent We remark: The requirement that $a(b-1)> 1$ excludes those
states studied in \cite{b2} which amount to $0<a<5$ and $b=(1+a)/a,$
whereby $a(b-1)=1.$
\section{A Non-polytropic Case}\label{sc8}
We now consider an example where $\Gamma$ is not necessarily
constant. We will consider a case where we replace $P(x)$ by
$P(\rho,x)$ where
\begin{equation} P(\rho,x)=T(x)\rho+L(x)\label{ex3}\end{equation}
for a non-increasing function $T$. The function $T$ can be
considered a function of temperature as in some physical models,
where temperature varies with position.

We compute $q_1(x)$ and $q_2(x)$ for $\rho$ as in Proposition
\ref{polytrop} and $P$ as in (\ref{ex3}) with
$$
T(x)=K_0\cdot(R_*-x)^{ab-a}\,\,\text{\rm and}\,\,
L(x)=L_0\cdot(R_*-x)^c
$$
for fixed $K_0,L_0>0,$ $a\geq 1,$ $b\geq1,$ and $c>0$. With
$w(x)=x^4\rho(x)$ and $p(x)$ $=$ $T(x)\rho(x)x^4$ we have that
$q_2(x)$ is the same as that of Proposition \ref{polytrop}, but with
$bK$ replaced by $K_0$. We compute
\begin{align}
q_1(x)&=-\frac{abK_0\cdot(R_*-x)^{-1-a+ab}+4cL_0\cdot(R_*-x)^{-1-a+c}}{x}\notag\\
q_2(x)&=-\frac{K_0R_*^2(R_*-x)^{-2-a+ab}\fQ(x/R_*)}{16x^2}\notag
\end{align}
where, of course, $\fQ$ is the polynomial given by (\ref{poly}).

We comment: With the EOS (\ref{ex3}) we admit some cases of
unbounded $Q(X)$ such as the case $c\leq ab-1$ whereby $q_0(x)$
$\asymp$ $(R_*-x)^{-1-a+c}$ near $x=R_*$ and, hence,
$\lim_{x\rightarrow R^-_*}q_0(x)$ $=-\infty$ for $c<a+1$. We can
treat such unbounded $q_0$ by applying results from \cite{st} (see
also \cite{wa}) for which we will find the following estimates
useful:
\begin{align}
Q_1^{\prime}(X)(x)=&
\left(\frac{q_1^{\prime}}{W}\right)(x)\asymp(R_*-x)^{{-2+(c-a)+(ab-a)/2}}
\label{abscon}
\\
Q_1^{\prime\prime}(X)(x)=&\left(\frac{\fder{x}Q_1^{\prime}(X)(x)}{W(x)}\right)\asymp(R_*-x)^{-3+(c-a)+(ab-a)}\notag\\
Q_2^{\prime}(X)(x)=&
\left(\frac{q_2^{\prime}}{W}\right)(x)\asymp(R_*-x)^{{-3+3(ab-a)/2}}\notag
\end{align}
near $R_*$. As for physical motivation, we note that when $b>1$ and
$c>a,$ we have that $P/\rho$ $\rightarrow 0$ as $x\rightarrow$
$R_*,$ corresponding to some related models as discussed in Section
8.3 of \cite{co}. Here pressure can be attributed to that of perfect
gas, given by $T\cdot\rho$ where temperature $T$ vanishes at the
surface, and another source, given by $L(x)$ - such as radiation,
for instance. In the case $c=a+1\leq ab$ we can also assign $L_0$ or
perhaps $K_0$ so that HSE holds at $x=R_*$. Moreover, we note that
our general results do not depend on the domain length $\gd>0$
where, by continuity arguments, (\ref{hsecont}) can be made
arbitrarily small.
\begin{thm}\label{nonpolyprop}
For the EOS (\ref{ex3}) with $a,b\geq 1$ suppose that either
\begin{itemize}
\item[i)] $ab\geq$ $a+2$ and
$c>$ $a+1$; or,
\item[ii)] $ab\geq$ $a+3$ and $c> a$.
\end{itemize}
Then the absolutely continuous spectrum contains $(0,\infty)$ and,
the differential equation (\ref{othermain}) has no solutions of
$L^2(dX)$ class near $X=\infty$ for any such $\gl$.
\end{thm}
\begin{proof}
We note that in either case i) or ii) the result follows by Theorem
\ref{maincont} when $c>ab-1.$  Moreover, for the remainder of the
proof we may suppose that $c\leq ab$ since the argument for $c> ab$
is the same as that for $c=ab$.

We start with case i). We first consider the subcase $c\geq$
$(ab+a)/2$ where we find $Q$ $\in$ $L^1(dX)$ near $\infty$ so that
we may apply Theorems 5 and 6 of \cite{st}. In the subcase $a+1<$
$c$ $\leq$ $(ab+a)/2$ we have
$\lim_{X\rightarrow\infty}Q(X)=\lim_{x\rightarrow R^-_*}q_0(x)=0$
and that $Q(X)\leq 0$ is bounded so that we may again apply Theorems
5 and 6 of \cite{st}.

We now consider case ii), needing only to suppose $c$ $\leq$ $a+1.$
From (\ref{abscon}) we find that $Q_2(X),Q^{\prime}(X)$ and
$Q^{\prime}_2(X)$ are locally absolutely continuous.  We also have
the following estimates near $x=R_*:$
\begin{align}\frac{1}{\sqrt{1-Q_1(X)(x)}}
W(x) &\asymp (R_*-x)^{-(ab-a)/2-(c-a)/2+1/2}&\notin
L^1;\label{spccs}
\\
\frac{Q_2^{\prime}(X)(x)}{\gl-Q_1(X)(x)}W(x)&\asymp (R_*-x)^{(ab-a)-(c-a)-2}&\in L^1\notag;\\
\frac{Q_1^{\prime\prime}(X)(x)}{(\gl-Q_1(X)(x))^{3/2}}W(x)&\asymp(R_*-x)^{(-3+2(ab-a)-(c-a))/2}&\in L^1\notag;\\
\frac{(Q_1^{\prime}(X)(x))^2}{(\gl-Q_1(X)(x))^{5/2}}W(x)&\asymp(R_*-x)^{(-3+(ab-a)-(c-a))/2}&\in
L^1. \notag
\end{align}
We may therefore apply Theorems 9 and 2 of \cite{st}.\end{proof}
\begin{remark}\label{remark}
We note the results of Theorem \ref{nonpolyprop} hold for other
combinations of $a,b>1$ and $c>a$ provided $ab-a> 2$ and the
exponent of the RHS of estimate (\ref{spccs}) is no greater than
$-1$.
\end{remark}

We analyze our model, demonstrating unbounded $\gd r$ while
satisfying the regularity condition via
\begin{thm}\label{lastthm}
The following hold for all non-trivial solutions in $\fT_{\gl}$ for
the equation as in Theorem \ref{nonpolyprop}:
\begin{itemize}
\item[i)] The regularity condition is satisfied; and,
\item[ii)] the corresponding $\gd r(r)$ is unbounded as $r\rightarrow R^-_*.$
\end{itemize}
\end{thm}
\begin{proof}
To prove item i) we use (\ref{y}) and apply estimates for $Y$ and
$Y^{\prime}$ according to various cases in the proof of Theorem
\ref{nonpolyprop}.

For $c>a+1$ we may apply Theorem \ref{B1} to obtain
$$Y(X),Y^{\prime}(X)\asymp 1$$ near $\infty$ and the result follows
as in Theorem \ref{polytrop}. For $a<c\leq a+1$ we have from Lemma
12 \cite{st} that
\begin{equation}
Y(X)\lesssim
\left(1+|Q_1(X)|\right)^{-1/4}\,\,;\,\,Y^{\prime}(X)\lesssim
\left(1+|Q_1(X)|\right)^{1/4}\notag
\end{equation}
near $\infty$ so that from (\ref{q}), the chain rule, and the
product rule
\begin{equation}
y(x)\lesssim
(R_*-x)^{\frac{1-ab-c}{4}}\,\,;\,\,y^{\prime}(x)\lesssim
(R_*-x)^{\frac{c-1-3ab}{4}}+(R_*-x)^{\frac{-c-ab-3}{4}}\notag
\end{equation} near $R_*$.
The limit (\ref{bdyc}) is obtained for $y=\xi$ by the estimates
\begin{align}
\left(P\Gamma\right)(x)&\lesssim (R_*-x)^{ab}\label{lastest}\\
\left(P\Gamma y\right)(x)&\lesssim (R_*-x)^{\frac{3(ab-a)+3a+1-c}{4}}\lesssim (R_*-x)^{2}\notag\\
\left(P\Gamma y^{\prime}\right)(x)&\lesssim
(R_*-x)^{\frac{c+ab-1}{4}}+(R_*-x)^{\frac{3ab-3-c}{4}} \lesssim
(R_*-x)^{3/4} \notag\end{align} which clearly vanish as
$x\rightarrow R_*^-.$

Item ii) follows from Theorem \ref{nonpolyprop} and Proposition
\ref{unbddosc}.
\end{proof}
\begin{remark}
In the estimates \ref{lastest} of the case $c<a+1$ we only suppose
$ab-a>2$ and, hence, the results of Theorem \ref{lastthm} likewise
hold for cases as discussed in Remark \ref{remark}.\end{remark}
\bibliographystyle{plain}

\end{document}